\documentclass[letterpaper, 10 pt, conference]{ieeeconf}  

\IEEEoverridecommandlockouts
\usepackage[english]{babel} 								
\usepackage[babel]{csquotes}
\usepackage{latexsym}												
\usepackage{url} 													
\usepackage{fancyhdr}										

\usepackage{emptypage}										
\usepackage[english]{varioref} 											
\usepackage{eurosym}										
\usepackage{lettrine}                                                                                 						\usepackage[fleqn]{amsmath}									                                                                                                                        \DeclareMathOperator{\diag}{diag}                                                           \usepackage{mathrsfs}										
\usepackage{amsfonts}										
\usepackage{amsmath, amssymb,amsfonts,dsfont}
\usepackage{etoolbox}
\usepackage{theorem}
\newtheorem{theorem}{Theorem}                						
\newtheorem{lemma}{Lemma}

\newtheorem{assumption}{Assumption}               						\newtheorem{remark}{Remark}
\newtheorem{objective}{Objective}
\usepackage{algorithm}
\usepackage{algpseudocode}
\usepackage{upgreek}										
\newcommand{\numberset}{\mathbb}							

\newcommand{\R}{\numberset{R}}
\usepackage{mathtools}										
								
\usepackage{graphics} 										
\usepackage{epsfig} 													
\usepackage{tabularx}										
\usepackage{longtable}										
\usepackage{pdflscape}
\usepackage{rotating}										
\usepackage{listings}										
\usepackage{xcolor} 										
\usepackage{siunitx}										\renewcommand{\vec}{\boldsymbol}								
					\DeclareMathAlphabet{\mathpzc}{OT1}{pzc}{m}{it}
\DeclareMathAlphabet{\mathcal}{OMS}{cmsy}{m}{n}

\usepackage{balance}
\usepackage{multirow}
\usepackage[american,cute inductors,smartlabels]{circuitikz}
\usepackage{pgf}
\usepackage{tikz}
\usepackage{cite}
\usepackage{tabularx}
\usepackage{longtable}
\usepackage{booktabs}
\usetikzlibrary{arrows,automata}

\usepackage{textcomp}

\title{\LARGE \bf
Robust Passivity-Based Control of Boost Converters in DC Microgrids$^\star$
}
\author{Michele Cucuzzella$^{1}$, Riccardo Lazzari$^{2}$, Yu Kawano$^{1}$, Krishna C. Kosaraju$^{1}$ and Jacquelien M.A. Scherpen$^{1}$
\thanks{$^\star$This work is supported by the EU Project \lq MatchIT' (project number 82203) and the Research Fund for the Italian Electrical System under the Contract Agreement between RSE S.p.A. and the Ministry of Economic Development - General Directorate for Nuclear Energy, Renewable Energy and Energy Efficiency in compliance with the Decree of March 8, 2006.}
\thanks{$^{1}$M. Cucuzzella, Y. Kawano, K.C. Kosaraju and J.M.A. Scherpen are with the Jan C. Wilems Center for Systems and Control, ENTEG, Faculty of Science and Engineering, University of Groningen, Nijenborgh 4, 9747 AG Groningen, the Netherlands 
        {\tt\small \{m.cucuzzella, y.kawano, k.c.kosaraju, j.m.a.scherpen\}@rug.nl}}%
\thanks{$^{2}$R. Lazzari is with the Department of Power Generation Technologies and Materials, RSE S.p.A., via Rubattino Raffaele 54, 20134 Milan, Italy
        {\tt\small Riccardo.Lazzari@rse-web.it}}%
}
\begin{document}
\maketitle
\thispagestyle{empty}
\pagestyle{empty}
\begin{abstract}
This work deals with the design of a robust and decentralized passivity-based control scheme for regulating the voltage of a DC microgrid through boost converters. 
A Krasovskii-type storage function is proposed and a (local) passivity property for DC microgrids comprising unknown \lq ZIP' (constant {impedance} \lq Z', constant {current} \lq I' and constant {power} \lq P') loads is established.
More precisely, the input port-variable of the corresponding passive map is equal to the first-time derivative of the control input. Then, the integrated input port-variable is used to shape the closed loop storage function such that it has a minimum at the desired equilibrium point. 
Convergence to the desired equilibrium is theoretically analyzed and the proposed control scheme is validated through experiments on a real DC microgrid.
\end{abstract}
%
%
%
%
\section{INTRODUCTION} 
Distributed Generation (DG) and the necessity of storing energy require fundamental transformations of the conventional power generation, transmission and distribution systems~\cite{ACKERMANN2001195}.
DG represents a conceptual solution to $i)$ enhance the integration of Renewable Energy Sources (RES) in order to reduce the dependency on fossil fuels and CO$_2$ emissions, $ii)$ increase the energy efficiency by reducing the transmission power losses, $iii)$  improve the service quality by enabling the operation of portions of the network  disconnected from the main grid and $iv)$ minimize the costs for electrifying remote areas or re-powering the existing power networks due to the ever increasing electric demand. 
A set of multiple DG Units (DGUs), loads and energy storage devices interconnected through power lines is identified in the literature as a \emph{microgrid} \cite{hatziargyriou2014microgrids}. 

In the last decades, due to the prevalence of Alternating Current~(AC) networks, the literature on microgrids mainly considered AC systems (see for instance~\cite{doi:10.1080/00207179.2015.1104555,7500071,GUI2018551,C_Cucuzzella_18_3} and the references therein).
However, the recent widespread use of RES as DGUs is motivating the design and operation of Direct Current (DC) microgrids~\cite{7268934}. 
Several devices (e.g. electric vehicles, electronic appliances, batteries and photovoltaic panels) can indeed be directly connected to a DC network avoiding lossy DC-AC conversion stages and the issues related to the frequency and reactive power control~\cite{PLANAS2015726}. Besides the development of industrial, commercial and residential DC distribution networks, some examples of existing or promising DC microgrid applications are ships, mobile military bases, trains, aircrafts and charging stations for electric vehicles.
For all these reasons, control of DC microgrids and, consequently, DC-DC power converters is gaining growing interest.

In DC microgrids, control schemes are usually designed to achieve voltage stabilization and current (or power) sharing (see for instance~\cite{DePersis2016,han8026170,cucuzzella2017robust,J_Trip2018,StrehlePfeiferKrebs2019_1000090285} and the references therein). 
However, the dynamics of the power converters are often neglected or described by linear models (e.g., buck converters).
Differently, in this letter we design a robust and decentralized passivity-based control scheme for regulating the voltage of a DC microgrid through \emph{boost} converters, the dynamics of which are \emph{nonlinear}. 
Regulating the voltage towards the nominal value is required to ensure a proper operation of the connected loads and guarantee the network stability.

\subsection{Literature Review and Main Contributions}
We provide now a brief comparison with some existing theoretical results dealing with the design of voltage controllers for boost converters. Simple tuning rules of passivity-preserving controllers are provided in \cite{1323174}, while stability in presence of bounded input is analyzed in \cite{Moreno2017}.
However, only constant impedance loads are considered, the network dynamics are neglected and in \cite{Moreno2017} the load resistance is assumed to be known. 
In \cite{7983406} and \cite{8618882}, Plug-and-Play voltage controllers are proposed. 
More precisely, the controller designed in \cite{7983406} is robust with respect to load uncertainties. However, the line dynamics are neglected and only local stability is established for the boost converter. In \cite{8618882} the microgrid stability is proved considering bounded input. However, only constant current loads are considered and the controller requires local information (including the load) and the value of the resistance of the lines interconnecting with the neighboring nodes. Under the assumption that the equilibrium point is known, a novel nonlinear control law that takes into account the constraints of the control action is proposed in \cite{Iovine}. 

We can now list the main contributions of this work:
\subsubsection{Nonlinear model} The considered microgrid model takes into account the nonlinear dynamics of boost converters and a possible meshed network topology, incorporating dynamic resistive-inductive lines and a general nonlinear load model (called \lq ZIP') including constant {impedance} \lq Z', constant {current} \lq I' and constant {power} \lq P'.
\subsubsection{Passivity framework} A Krasovskii-type storage function \cite{krasovskiicertain, 7846443} is proposed and a (local) passivity property for the considered DC microgrid is established.
More precisely, the input port-variable of the corresponding passive map is equal to the first-time derivative of the control input. Then, the integrated input port-variable is used to shape (\emph{input shaping} methodology~\cite{2018arXiv181102838C}) the closed loop storage function such that it has a minimum at the desired equilibrium point. Convergence to the desired equilibrium is established together with extremely simple tuning rules.
\subsubsection{Robustness} The proposed control scheme is decentralized and robust with respect to unknown loads and other parameter uncertainty (e.g., line and filter impedances).
\subsubsection{Validation} The proposed control strategy is verified through experiments on a real DC microgrid test facility at Ricerca sul Sistema Energetico (RSE), Milan, Italy, showing excellent closed-loop performance (see~\cite{7892743} and \cite{J_Cucuzzella_18} for more information about the experimental setup where we have performed our tests).

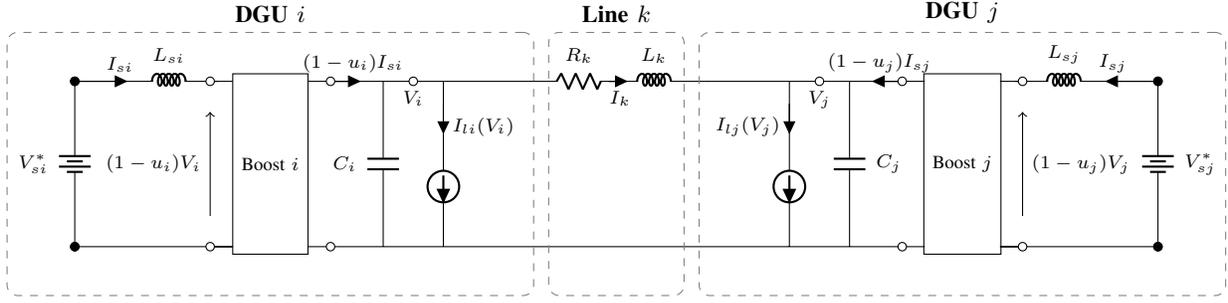
\begin{figure*}[t]
\begin{center}
\ctikzset{bipoles/length=0.7cm}
\begin{circuitikz}[scale=1, transform shape]
\ctikzset{current/distance=1}
\draw
node[] (Ti) at (0,0) {}
node[] (Tj) at ($(5.4,0)$) {}
node[ocirc] (Aibattery) at ([xshift=-4.5cm,yshift=1.1cm]Ti) {}
node[ocirc] (Bibattery) at ([xshift=-4.5cm,yshift=-1.1cm]Ti) {}
(Aibattery) to [battery, l_=\scriptsize{$V^\ast_{si}$},*-*] (Bibattery) {}
node [rectangle,draw, minimum width=1cm,minimum height=2.4cm] (boosti) at ($0.5*(Aibattery)+0.5*(Bibattery)+(2.6,0)$) {\scriptsize{Boost $i$}}
node[ocirc] (Ai) at ($(Aibattery)+(1.8,0)$) {}
node[ocirc] (Bi) at ($(Bibattery)+(1.8,0)$) {}
(Ai) to [short] ([xshift=0.3cm]Ai)
(Bi) to [short] ([xshift=0.3cm]Bi)
(Aibattery) 
to [short,i=\scriptsize{$I_{si}$}]($(Aibattery)+(0.6,0)$){}
to [L, l=\scriptsize{$L_{si}$}]($(Aibattery)+(1.9,0)$){}
(Bibattery) to [short] ([xshift=2.1cm]Bibattery)
node[ocirc] (AAi) at ($(Ai)+(1.6,0)$) {}
node[ocirc] (BBi) at ($(Bi)+(1.6,0)$) {}
(AAi) to [short] ([xshift=-0.3cm]AAi)
(BBi) to [short] ([xshift=-0.3cm]BBi)
(AAi) to [short,i=\scriptsize{$(1-u_i)I_{si}$}]($(AAi)+(0.3,0)$){}
(AAi) to [short] ([xshift=0.7cm]AAi)
(BBi) to [short] ([xshift=0.7cm]BBi);
\begin{scope}[shorten >= 10pt,shorten <= 10pt,]
\draw[<-] (Ai) -- node[left] {\scriptsize{$(1-u_{i})V_i$}} (Bi);
\end{scope};
\draw
($(Ti)+(0.0,1.1)$) node[anchor=north]{\scriptsize{$V_{i}$}}
($(Ti)+(0.0,1.1)$) node[ocirc](PCCi){}
($(Ti)+(0.4,1.1)$)--($(Ti)+(0.4,0.8)$) to [short,i>=\scriptsize{$I_{li}(V_i)$}]($(Ti)+(0.4,0.5)$)
to [I]($(Ti)+(0.4,-1.1)$)
($(Ti)+(-0.4,1.1)$) to [C, l_=\scriptsize{$C_{i}$}] ($(Ti)+(-0.4,-1.1)$)
($(Ti)+(-0.4,1.1)$) to [short] ($(Ti)+(1.5,1.1)$)
($(Ti)+(2.6,1.1)$) to [short,i_=\scriptsize{$I_{k}$}] ($(Ti)+(2.75,1.1)$)--($(Ti)+(2.7,1.1)$)
($(Ti)+(1.5,1.1)$)--($(Ti)+(1.9,1.1)$) to [R, l=\scriptsize{$R_{k}$}] ($(Ti)+(2.5,1.1)$) {}
to [L, l={\scriptsize{$L_{k}$}}, color=black]($(Tj)+(-1.5,1.1)$){}
to ($(Tj)+(-1.5,1.1)$){}
($(Tj)+(-1.5,1.1)$) to [short] ($(Tj)+(-1.2,1.1)$)--($(Tj)+(-0.8,1.1)$)
($(Tj)+(-0.8,1.1)$) to [short] ($(Tj)+(0.4,1.1)$)
($(Ti)+(-0.4,-1.1)$) to [short] ($(Tj)+(0.4,-1.1)$)
($(Tj)+(-0,1.1)$) node[anchor=north]{\scriptsize{$V_{j}$}}
($(Tj)+(-0,1.1)$) node[ocirc](PCCj){}
($(Tj)+(-0.4,1.1)$)--($(Tj)+(-0.4,0.8)$) to [short,i>_=\scriptsize{$I_{lj}(V_j)$}]($(Tj)+(-0.4,0.5)$)
to [I]($(Tj)+(-0.4,-1.1)$)
($(Tj)+(0.4,1.1)$) to [C, l=\scriptsize{$C_{j}$}] ($(Tj)+(0.4,-1.1)$)
node[ocirc] (Ajbattery) at ([xshift=4.5cm,yshift=1.1cm]Tj) {}
node[ocirc] (Bjbattery) at ([xshift=4.5cm,yshift=-1.1cm]Tj) {}
(Ajbattery) to [battery, l=\scriptsize{$V^\ast_{{sj}}$},*-*] (Bjbattery){}
node [rectangle,draw,minimum width=1cm,minimum height=2.4cm] (boostj) at ($0.5*(Ajbattery)+0.5*(Bjbattery)-(2.6,0)$) {\scriptsize{Boost $j$}}
node[ocirc] (Aj) at ($(Ajbattery)+(-1.8,0)$) {}
node[ocirc] (Bj) at ($(Bjbattery)+(-1.8,0)$) {}
(Aj) to [short] ([xshift=-0.3cm]Aj)
(Bj) to [short] ([xshift=-0.3cm]Bj)
(Ajbattery) to [short,i_=\scriptsize{$I_{sj}$}]($(Ajbattery)+(-0.6,0)$){}
($(Tj)+(2.6,1.1)$) to [L, l=\scriptsize{$L_{sj}$}]($(Ajbattery)+(-0.6,0)$){}
(Bjbattery) to [short] ([xshift=-2.1cm]Bjbattery)
node[ocirc] (AAj) at ($(Aj)+(-1.6,0)$) {}
node[ocirc] (BBj) at ($(Bj)+(-1.6,0)$) {}
(AAj) to [short] ([xshift=0.3cm]AAj)
(BBj) to [short] ([xshift=0.3cm]BBj)
(AAj) to [short,i_=\scriptsize{$(1-u_j)I_{sj}$}]($(AAj)+(-0.3,0)$){}
(AAj) to [short] ([xshift=-0.7cm]AAj)
(BBj) to [short] ([xshift=-0.7cm]BBj);
\begin{scope}[shorten >= 10pt,shorten <= 10pt,]
\draw[<-] (Aj) -- node[right] {\scriptsize{$(1-u_{j})V_j$}} (Bj);
\end{scope};
\draw
node [rectangle,draw,minimum width=7.0cm,minimum height=3.5cm,dashed,rounded corners,color=gray,label=\small\textbf{DGU $i$}] (DGUi) at ($0.5*(Aibattery)+0.5*(Bibattery)+(2.6,0)$) {}
node [rectangle,draw,minimum width=7.0cm,minimum height=3.5cm,dashed,rounded corners,color=gray,label=\small\textbf{DGU $j$}] (DGUj) at ($0.5*(Ajbattery)+0.5*(Bjbattery)-(2.6,0)$) {}
node [rectangle,draw,minimum width=1.8cm,minimum height=3.5cm,dashed,rounded corners,color=gray,label=\small\textbf{Line $k$}] (Lineij) at ($0.5*(DGUi.center)+0.5*(DGUj.center)+(0,0)$){};
\end{circuitikz}
\caption{Electrical scheme of a typical boost-based DC microgrid composed of two DGUs connected by a line.}
\label{fig:microgrid1}
\end{center}
\end{figure*}

\subsection{Outline}
The present letter is organized as follows. The microgrid model is described in Section~\ref{sec:preliminaries}, while the control objective is formulated in Section~\ref{sec:formulation}. In Section~\ref{sec:control}, the proposed control scheme is designed and the stability of the controlled microgrid analyzed. In Section~\ref{sec:experiments}, the proposed control scheme is validated through experiments on a real DC microgrid and, finally, conclusions are gathered in Section VI.
\subsection{Notation}
Let $\vec{0}$ be the vector of all zeros of suitable dimension and let $\mathds{1}_n \in \R^n$ be the vector containing all ones. The $i$-th element of vector $x$ is denoted by $x_i$. 
 A steady state solution to system $\dot x = \zeta(x)$, is denoted by $\overline x$,  i.e., $\boldsymbol{0} = \zeta(\overline x)$. 
 A constant signal is denoted by $x^\ast$. Given a vector $x\in \R^n$, $[x]\in \R^{n\times n}$ indicates the diagonal matrix whose diagonal entries are the components of $x$. Let \lq$\circ$' denote the Hadamard product, i.e., given vectors $ x, y \in \R^n$, $( x \circ  y) \in \R^n$ is a vector with elements $( x \circ  y)_i := x_i y_i$ for all $i=1,\dots,n$. 
 Let $A \in \R^{n \times n}$ be a matrix. In case $A$ is a positive definite (positive semi-definite) matrix, we write $A \succ 0$  ($A \succeq 0$). The $n \times n$ identity matrix is denoted by $\mathds{I}_n$.
 Given a set $\Omega$, $|\Omega|$ represents the cardinality of $\Omega$.
\section{DC Microgrid Model}
\label{sec:preliminaries}

The DC microgrid is represented by a connected and undirected graph $\mathcal{G} = (\mathcal{V},\mathcal{E})$, where   $\mathcal{V} = \{1,...,n\}$ is the set of nodes and $\mathcal{E} \subseteq \mathcal{V}\times\mathcal{V}$ represents the set of the resistive-inductive lines interconnecting the nodes.
Each node, which we call \emph{Distributed Generation Unit}\footnote{Note that, we consider generation units only for the sake of simplicity and without loss of generality. In the experiments (see Section \ref{sec:experiments}), the controlled nodes are indeed storage units, i.e., batteries.} (DGU), includes a DC-DC boost converter supplying an unknown load.
A schematic electrical diagram of the considered DC network including two DGUs interconnected by a power line  is illustrated in Fig.~\ref{fig:microgrid1} (see also Table~\ref{tab:symbols} for the description of the used symbols).

\begin{table}
\begin{center}
\caption{Description of the used symbols}\label{tab:symbols}
\begin{tabular}{cl}
\toprule
&{\bf State variables}\\
\midrule
$I_s$						& Generated current\\
$V$						& Load voltage\\
$I$ 						& Line current \\
\midrule
&{\bf Inputs}\\
\midrule
$u$						& Control input (duty cycle) \\
$V^\ast_s$					& Voltage source\\
\midrule
&{\bf ZIP loads}\\
\midrule
$G^\ast_l$						& Load conductance\\
$I^\ast_l$						& Load current \\
$P^\ast_l$						& Load power \\
\midrule
&{\bf Filter and line parameters}\\
\midrule
$L_s$						& Boost filter inductance\\
$C$						& Boost filter capacitor\\
$R$						& Line resistance\\
${L}$						& {Line inductance}\\
\bottomrule
\end{tabular}
\end{center}
\end{table}

By applying the Kirchhoff's laws, the \emph{average}\footnote{Under the condition that the Pulse Width Modulation (PWM) frequency is sufficiently high, the state of the system can be replaced by the average state representing the average inductor currents and capacitor voltages. Consequently, the switching control input is replaced by the so-called duty cycle of the converter.} governing dynamic equations\footnote{For the sake of simplicity, the dependence of all the variables on time $t$ is omitted when it is clear from the context.} of the node $i \in \mathcal{V}$ are the following:
\begin{align}
\begin{split}
\label{eq:plant_i1}
L_{si}\dot{I}_{si} &= - \left(1-u_i\right) V_i + V^\ast_{si} \\
C_{i}\dot{V}_i &=  \left(1-u_i\right) I_{si} - {I}_{li}(V_i) -  { \sum_{k \in \mathcal{E}_i}^{}I_{k}},
\end{split}
\end{align}
where $I_{si}: \R_{\geq0} \rightarrow \R, I_k: \R_{\geq0} \rightarrow \R, V_i:\R_{\geq0} \rightarrow \R_{>0}, I_{li}(V_i):\R_{>0} \rightarrow \R_{\geq0}, u_i : \R_{\geq0} \rightarrow [0,1)$ and $V^\ast_{si}, L_{si}, C_i \in \R_{>0}$. Moreover, $\mathcal{E}_i$ is the set of power lines connected to the DGU~$i$ and $I_{k}$ is the current flowing on the line $k \in \mathcal{E}_i$.
Let $k$ be the power line interconnecting DGUs $i, j \in \mathcal{V}$. Then, the dynamic of $I_{k}$ in~\eqref{eq:plant_i1} is given by
\begin{equation}
\label{eq:plant_i2}
L_{k}{\dot I_{k}} = \left(V_i - V_j\right) - R_{k} I_{k},
\end{equation}
with $L_{k},R_k \in \R_{>0}$.
Moreover, the term $I_{li}(V_i)$ in \eqref{eq:plant_i1} represents the current demanded\footnote{The results presented in this work hold also in case of the so-called \emph{net generating} loads, i.e., $I_{li}(V_i)<0$.} by the load $i \in \mathcal{V}$ and (generally) depends on the node voltage $V_i$. 
In this work, we consider a general load model including the parallel combination of the following load components:
\begin{enumerate}
\item{constant impedance: $I_{li} = G^\ast_{li}V_i$, with $G^\ast_{i}\in\R_{>0}$,}
\item{constant current: $I_{li} = I^\ast_{li}$, with $I^\ast_{li}\in\R_{\geq0}$, and}
\item{constant power: $I_{li} = V_i^{-1}P^\ast_{li}$, with $P^\ast_{li}\in\R_{\geq0}$.}
\end{enumerate}
To refer to the load types above, the letters \lq Z', \lq I' and \lq P', respectively, are often used in the literature \cite{DePersis2016}.
Therefore, in presence of the so-called ZIP loads, $I_{li}(V_i)$ in \eqref{eq:plant_i1} is given by
\begin{equation}
\label{eq:plant_i3}
I_{li}(V_i) = G^\ast_{li}V_i + I^\ast_{li} + V_i^{-1}P^\ast_{li}.
\end{equation}
The symbols used in  \eqref{eq:plant_i1}--\eqref{eq:plant_i3} are described in Table~\ref{tab:symbols}.

We represent the microgrid topology by using its corresponding incidence matrix $ D \in \R^{n \times |\mathcal{E}|}$. The ends of edge $k \in \mathcal{E}$ are arbitrarily labeled with a $+$ and a $-$. More precisely, one has that
\begin{equation*}
\label{eq:incidence}
D_{ik}=
\begin{cases}
+1 \quad &\text{if $i$ is the positive end of $k$}\\
-1 \quad &\text{if $i$ is the negative end of $k$}\\
0 \quad &\text{otherwise}.
\end{cases}
\end{equation*}
The overall microgrid system \eqref{eq:plant_i1}, \eqref{eq:plant_i2} in presence of ZIP loads \eqref{eq:plant_i3} can now be written compactly for all nodes $i \in \mathcal{V}$ as follows:
\begin{align}
\label{eq:plant}
	\begin{split}
		 {L_s} {\dot{I}_{s}} &= - \left(\mathds{1}_n - u\right) \circ  {V} +  {V^\ast_{s}}\\
		 {C} {\dot{V}} &=   \left(\mathds{1}_n - u\right) \circ  {I_s} - G^\ast_l V - I^\ast_l - [V]^{-1}P^\ast_l + D I\\
		 L \dot I &= -D^\top V - RI,
	\end{split}
\end{align}
where ${I_s}: \R_{\geq0} \rightarrow \R^n$, $V:\R_{\geq0} \rightarrow \R^n_{>0}$, $I: \R_{\geq0} \rightarrow \R^{|\mathcal{E}|}$, $u:\R_{\geq0} \rightarrow [0,1)^n$, $V^\ast_{s}\in \R^n_{>0}$ and $I^\ast_l, P^\ast_l \in \R^n_{\geq0}$. Moreover, the matrices $ {L_s},  {C}, G^\ast_l, L$ and $R$ have appropriate dimensions and are constant, positive definite and diagonal.
\section{Problem Formulation: Voltage Regulation}
\label{sec:formulation}
In this section, we formulate the control objective aiming at regulating the voltage of a boost-based DC microgrid.
First, we notice that for given $u^\ast, V^\ast_s, G^\ast_l, I^\ast_l$ and $P^\ast_l$, a steady state solution $(\overline I_s, \overline V, \overline I)$ to system~\eqref{eq:plant} satisfies
\begin{subequations}\label{eq:steady_state}
\begin{align}
		\overline{V} &=  \left( \mathds{I}_n -[u^\ast] \right)^{-1} V_s^\ast \label{eq:ss1}\\
		 \left(\mathds{I}_n - [u^\ast] \right) \overline I_s &= G^\ast_l \overline V + I^\ast_l + [\overline V]^{-1}P^\ast_l - D \overline I \label{eq:ss2}\\
		 \overline{I} &= -R^{-1}D^\top\overline V \label{eq:ss3},
\end{align}
\end{subequations}
where from \eqref{eq:ss1} it follows that the boost output voltage $\overline V_i$ is higher than the  voltage source $V^\ast_{si}, i \in \mathcal{V}$, while \eqref{eq:ss2} implies\footnote{The incidence matrix $D$, satisfies $\mathds{1}_n^\top D = \boldsymbol{0}$.} that \emph{current balance} is achieved at the steady state, i.e., the total current $\mathds{1}^\top  \left(\mathds{I}_n - [u^\ast] \right) \overline I_s$ injected by the boost converters is equal to the total current $\mathds{1}^\top(G^\ast_l \overline V + I^\ast_l + [\overline V]^{-1}P^\ast_l)$ demanded by the ZIP loads.
Moreover, in order to guarantee a proper functioning of the connected loads, it is required that the current balance is achieved at the desired voltage value.
Consequently, before formulating the control objective, we introduce the following assumption on the existence of a desired reference voltage for each DGU:
\begin{assumption}{\bf(Desired voltage)}
\label{ass:desired_voltages}
There exists a constant desired reference voltage $V_{di}^{\ast} \geq V_{si}^\ast$ for all $i \in \mathcal{V}$.
\end{assumption}

Given ${V}^\ast_d = [V_{d1}^\ast, \dots, V_{dn}^\ast]^T$, the control objective is then formulated as follows:
\begin{objective}{\bf(Voltage regulation)}
\label{obj:voltage_regulation}
\begin{equation*}
\label{eq:voltage_balancing}
\lim_{t \rightarrow \infty}   V(t) =  \overline{  V} =   V^{\ast}_d.
\end{equation*}
\end{objective}
Moreover, in order to permit the controller design in the next section, the following assumption is introduced on the available informations:

\begin{assumption}{\bf(Available informations)}
\label{ass:available_information}
The state variables $I_{si}$, $V_i$ and the voltage source $V^\ast_{si}$ are locally available only at the DGU $i$. 
\end{assumption}
Consequently, the control scheme we design in Section~\ref{sec:control} to achieve Objective 1 needs to be fully \emph{decentralized}, increasing the practical applicability of the proposed approach.

\begin{remark}{\bf(Microgrid uncertainty)}
\label{rmk:varying_uncertainty_current_demand}
Note that, according to Assumption~\ref{ass:available_information}, the parameters $I^\ast_l, P^\ast_l, G^\ast_l, L_s, L, C, R$ of the ZIP loads, lines and boost converters are not known. As a consequence, we need to design a control scheme that is robust with respect to the system uncertainty.
\end{remark}
\section{The Proposed Solution}
\label{sec:control}
In this section, we introduce the key aspects of the proposed decentralized passivity-based control scheme aiming at achieving Objective 1. More precisely, we first augment system \eqref{eq:plant} with additional dynamics. Secondly, we propose a Krasovskii-type storage function \cite{krasovskiicertain, 7846443} and establish a (local) passivity property for the augmented system. The input port-variable of the corresponding passive map is equal to the first-time derivative of the control input. Then, we use the integrated input port-variable to shape the closed loop storage function such that it has a minimum at the desired equilibrium point.

Consider the following \emph{auxiliary}\footnote{The state variables and the control input of the \emph{auxiliary} system are $I_s, V, I, \dot{I}_s, \dot V, \dot I, u$ and $\upsilon_c$, respectively.} system:
\begin{subequations}
\label{eq:plant_extended}
\begin{align}
		{L_s} {\dot{I}_{s}} = &- \left(\mathds{1}_n - u\right) \circ  {V} +  {V^\ast_{s}} \label{eq:pe1}\\
		 {C} {\dot{V}} =   &\left(\mathds{1}_n - u\right) \circ  {I_s} - G^\ast_l V - I^\ast_l - [V]^{-1}P^\ast_l + D I \label{eq:pe2}\\
		 L \dot I = &-D^\top V - RI \label{eq:pe3}\\	
		 {L_s} {\ddot{I}_{s}} = &- \left(\mathds{1}_n - u\right) \circ  {\dot V} +  \upsilon_c \circ V \label{eq:pe4}\\
		 {C} {\ddot{V}} =  &\left(\mathds{1}_n - u\right) \circ  {\dot I_s} - \upsilon_c \circ I_s \nonumber\\ &- \left( G^\ast_l - [V]^{-2}[P^\ast_l] \right) \dot V + D \dot I \label{eq:pe5}\\
		 L \ddot I = &-D^\top \dot V - R \dot I \label{eq:pe6}\\	
		\dot u = &~\upsilon_c, \label{eq:pe7}
\end{align}
\end{subequations}
which includes also the dynamics of the first-time derivative of the state and input of system \eqref{eq:plant}. 

Let the vector $z := (I_s^\top, V^\top, I^\top, \dot{I}_s^\top, \dot V^\top, \dot I^\top, u^\top)^\top \in \mathcal{Z} := \left\{ z \in \R^{5n+2|\mathcal{E}|}: V \in \R^n_{>0}, u \in [0,1)^n\right\}$ denote the state of the auxiliary system \eqref{eq:plant_extended}.
In order to establish a passivity property for system \eqref{eq:plant_extended}, we first introduce the following set:
\begin{equation*}
{\mathcal{Z}}_{\mathrm{ZIP}} := \left\{ z \in \mathcal{Z} : G_l^\ast - [V]^{-2}[P_l^\ast] \succ 0  \right\}.
\end{equation*}
Then, the following result can be proved.

\begin{lemma}{\bf(Passivity property)}
\label{lemma:1}
System~\eqref{eq:plant_extended} is passive with respect to the supply rate $\upsilon_c^\top \left( \dot I_s \circ V - \dot V \circ I_s \right)$ and the storage function
\begin{equation}
\label{eq:S1}
S(\dot{I}_s, \dot V, \dot I) = \frac{1}{2}\dot{I}_s^\top L_s \dot{I}_s + \frac{1}{2}\dot{V}^\top C \dot{V} + \frac{1}{2}\dot{I}^\top L \dot{I},
\end{equation}
for all the trajectories $z \in {\mathcal{Z}}_{\mathrm{ZIP}}$.
\end{lemma}

\begin{proof}
The storage function $S$ in \eqref{eq:S1} satisfies
\begin{align*}
\begin{split}
\dot S &= - \dot V^\top \left( G^\ast_l - [V]^{-2}[P^\ast_l] \right) \dot V - \dot I^\top R \dot I \\
& \hspace{4mm} + \upsilon_c^\top \left( \dot I_s \circ V - \dot V \circ I_s \right)\\
 &\leq \upsilon_c^\top \left( \dot I_s \circ V - \dot V \circ I_s \right),
\end{split}
\end{align*}
along the solutions $z \in {\mathcal{Z}}_{\mathrm{ZIP}}$ to system~\eqref{eq:plant_extended}, which concludes the proof.
\end{proof}
\begin{remark}{\bf (Insights on the proposed storage function $S$)}
\label{rm:S}
The storage function $S$ in \eqref{eq:S1} depends on the states $\dot{I}_s, \dot V, \dot I$. This implies that $S$ depends also on $I_s, V, I$ and $u$, i.e., the entire state of the auxiliary system~\eqref{eq:plant_extended}. This is evident from replacing $\dot{I}_s, \dot V, \dot I$ by the corresponding dynamics \eqref{eq:pe1}--\eqref{eq:pe3}, or rewriting $S$ as follows:
\begin{align*}
\begin{split}
S(z) =&~\frac{1}{4}\left( \dot{I}_s^\top L_s \dot{I}_s + \dot{V}^\top C \dot{V} + \dot{I}^\top L \dot{I}\right)\\
&+\frac{1}{4}\left( f_{{I}_s}^\top L^{-1}_s f_{{I}_s} + f_{V}^\top C^{-1} f_{V} + f_{I}^\top L^{-1} f_{I}\right),
\end{split}
\end{align*}
where $f_{{I}_s}: \R^n_{>0} \times [0,1)^n \rightarrow \ \R^n$, $f_{V}: \R^n \times \R^n_{>0} \times \R^{|\mathcal{E}|} \times [0,1)^n \rightarrow \ \R^n$, and $f_{I}: \R^n_{>0} \times \R^{|\mathcal{E}|} \rightarrow \ \R^{|\mathcal{E}|}$ represent the right-hand sides of \eqref{eq:pe1}--\eqref{eq:pe3}, respectively.
Moreover, it will be shown in Theorem~\ref{th:1} that using \eqref{eq:S1} to design the controller permits, differently from \cite{Moreno2017} and \cite{8618882}, the achievement of Objective~1 despite the system uncertainty (see Remark~\ref{rmk:varying_uncertainty_current_demand}). However, the cost of designing a robust controller is the need of information about the first-time derivative of the signals $I_s$ and $V$ (see Remark~\ref{rm:robustness}).
\end{remark}
\smallskip
Before designing the controller and introducing the main result of this work, to be able to achieve Objective 1, we show that a unique steady state solution to system~\eqref{eq:plant_extended}
always exists.
\smallskip
\begin{lemma}{\bf(Existence of a unique steady state solution)}
\label{lemma:ss_solution}
Let Assumption \ref{ass:desired_voltages} hold. Given $\upsilon_c = 0$ and $V^\ast_{di} > \sqrt{P^\ast_{li}/G^\ast_{li}}$, for all $i \in \mathcal{V}$, there exists a unique steady state solution $\overline z = (\overline I_s,  \overline V = V^\ast_d, \overline I, \vec0, \vec0, \vec0, \overline u) \in {\mathcal{Z}}_{\mathrm{ZIP}}$ to system~\eqref{eq:plant_extended}, satisfying
\begin{align}
\label{eq:plant_extended_ss}
	\begin{split}
		\overline u = &~\mathds{1}_n -[V^\ast_d]^{-1} V_s^\ast\\
		 \left(\mathds{I}_n - [\overline u] \right) \overline I_s = &~G^\ast_l  V^\ast_d + I^\ast_l + [V^\ast_d]^{-1}P^\ast_l - D \overline I\\
		 \overline{I} = & -R^{-1}D^\top V^\ast_d\\	
		 \vec0 = &~{\dot V} \\
		 \vec0 = &~{\dot I_s}\\
		 \vec0 = &~\dot I\\	
		\vec0 = &~\upsilon_c,
	\end{split}
\end{align}
\end{lemma}
\begin{proof}
The proof follows from setting the left-hand-side of system \eqref{eq:plant_extended} to zero.
\end{proof}
\smallskip
We can now show the main result of this paper concerning the design of a controller that (provably) stabilizes system~\eqref{eq:plant_extended} achieving Objective 1.

\begin{figure*}
\begin{center}
\ctikzset{bipoles/length=0.7cm}
\begin{circuitikz}[scale=1, transform shape]
\ctikzset{current/distance=1}
\draw
node[] (Ti2) at (0,0) {}
node[] (Tj2) at ($(5.4,0)$) {}
node[ocirc] (Aibattery2) at ([xshift=-4.5cm,yshift=1.1cm]Ti2) {}
node[ocirc] (Bibattery2) at ([xshift=-4.5cm,yshift=-1.1cm]Ti2) {}
(Aibattery2) to [battery, l_=\scriptsize{$V^\ast_{s2}$},*-*] (Bibattery2) {}
node [rectangle,draw, minimum width=1cm,minimum height=2.4cm] (boosti2) at ($0.5*(Aibattery2)+0.5*(Bibattery2)+(2.5,0)$) {}
node (Ai2) at ($(Aibattery2)+(2.2,0)$) {}
node (Bi2) at ($(Bibattery2)+(2.2,0)$) {}
(Aibattery2) ($(Aibattery2)+(0.0,0)$) {}
to [short,i=\scriptsize{$I_{s2}$}]($(Aibattery2)+(0.6,0)$){}
to [L, l=\scriptsize{$L_{s2}$}]($(Aibattery2)+(2,0)$){}
(Bibattery2) to [short] ([xshift=2.0cm]Bibattery2)
node (AAi2) at ($(Ai2)+(1.6,0)$) {}
node (BBi2) at ($(Bi2)+(1.6,0)$) {}
(AAi2) to [short] ([xshift=0.5cm]AAi2)
(BBi2) to [short] ([xshift=0.7cm]BBi2);
\begin{scope}[shorten >= 10pt,shorten <= 10pt,]
\end{scope};
\draw
($(Ti2)+(-0.5,1.1)$) to [short] ($(Ti2)+(-1.5,1.1)$)
($(Ti2)+(-0.5,-1.1)$) to [short] ($(Ti2)+(-1.5,-1.1)$)
($(Ti2)+(-0.5,1.1)$) node[anchor=south]{\scriptsize{\color{black}{$V_2$}}}
($(Ti2)+(-0.5,1.1)$) node[ocirc, color=black](PCCi2){}
($(AAi2)+(0.5,0)$) to [R, l=\scriptsize{$R_{{12}}$}] ($(AAi2)+(1.6,0)$) {}
to [L, l=\scriptsize{$L_{{12}}$}]($(AAi2)+(2.3,0)$){}
($(AAi2)+(2.3,0)$) to [short] ($(AAi2)+(2.7,0)$)
($(AAi2)+(2.7,0)$) node[anchor=south]{\scriptsize{\color{black}{$V_1$}}}
($(AAi2)+(2.7,0)$) node[ocirc, color=black](node1){}
to [I, l_=\scriptsize{$I_{l1}(V_1)$}]($(Ti2)+(2,-1.1)$){}
($(Ti2)+(-0.5,1.1)$) to [C, l_=\scriptsize{$C_{2}$}] ($(Ti2)+(-0.5,-1.1)$)
($(AAi2)+(3,0)$) to [R, l=\scriptsize{$R_{{13}}$}] ($(AAi2)+(4.1,0)$) {}
to [L, l=\scriptsize{$L_{{13}}$}]($(AAi2)+(4.8,0)$){}
($(AAi2)+(4.8,0)$) to [short] ($(AAi2)+(5.2,0)$)
($(AAi2)+(3,0)$) to [short] ($(AAi2)+(2.7,0)$)
($(AAi2)+(5.2,0)$) node[anchor=south]{\scriptsize{\color{black}{$V_3$}}}
($(AAi2)+(5.2,0)$) node[ocirc, color=black](node3){}
(node3) to [short] ($(Tj2)+(-0.6,1.1)$)
($(Ti2)+(0,-1.1)$) to [short] ($(Tj2)+(0,-1.1)$)
($(Tj2)+(1.5,1.1)$) node[anchor=south]{\scriptsize{\color{black}{$V_4$}}}
($(Tj2)+(1.5,1.1)$) node[ocirc, color=black](node4){}
 (node3) to [I, l=\scriptsize{$I_{l3}(V_3)$}] ($(Tj2)+(-0.9,-1.1)$)
(node3) to [R, l=\scriptsize{$R_{{34}}$}] ($(node3)+(1.4,0)$) {}
to [L, l=\scriptsize{$L_{{34}}$}]($(node3)+(1.8,0)$){}
($(node3)+(1.8,0)$) to [short] (node4)
($(Tj2)+(1.5,-1.1)$) to [short] ($(Tj2)+(0,-1.1)$)
($(Tj2)+(1.5,-1.1)$) to [short] ($(Tj2)+(2.5,-1.1)$)
(node4) to [short] ($(node4)+(1,0)$)
($(Tj2)+(1.5,1.1)$) to [C, l=\scriptsize{$C_{4}$}] ($(Tj2)+(1.5,-1.1)$)
node[ocirc] (Ajbattery2) at ([xshift=5.5cm,yshift=1.1cm]Tj2) {}
node[ocirc] (Bjbattery2) at ([xshift=5.5cm,yshift=-1.1cm]Tj2) {}
(Ajbattery2) to [battery, l=\scriptsize{$V^\ast_{s4}$},*-*] (Bjbattery2){}
node [rectangle,draw,minimum width=1cm,minimum height=2.4cm] (boostj2) at ($0.5*(Ajbattery2)+0.5*(Bjbattery2)-(2.5,0)$) {\scriptsize{}}
node (Aj2) at ($(Ajbattery2)+(-2.2,0)$) {}
node (Bj2) at ($(Bjbattery2)+(-2.2,0)$) {}
(Ajbattery2) to [short,i_=\scriptsize{$I_{s4}$}]($(Ajbattery2)+(-0.6,0)$){}
($(Tj)+(3.5,1.1)$) to [L, l=\scriptsize{$L_{s4}$}]($(Ajbattery2)+(-0.6,0)$){}
(Bjbattery2) to [short] ([xshift=-2.0cm]Bjbattery2);
\draw (-2.0,0) node {\scriptsize Boost};
\draw (8.4,0) node {\scriptsize Boost};
\end{circuitikz}
\caption{Electrical scheme of the RSE's DC microgrid.}
\label{fig:microgrid2}
\end{center}
\end{figure*}
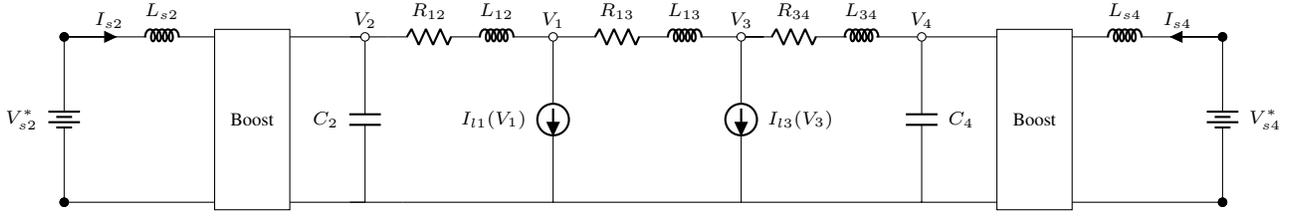

\begin{theorem}{\bf(Stability)}
\label{th:1}
Let Assumptions \ref{ass:desired_voltages}-\ref{ass:available_information} hold.
Consider system \eqref{eq:plant_extended} controlled by
\begin{equation}
\label{eq:controller}
T_c \upsilon_c = - K_c \left( u - u^\ast_d \right) - \left( \dot I_s \circ V - \dot V \circ I_s  \right),
\end{equation}
where $u^\ast_{di} = 1-V^\ast_{si}/V^\ast_{di}$ is the desired value of the duty cycle of the boost converter $i \in \mathcal{V}$, $T_c = \diag (T_{c1}, \dots, T_{cn})$, $K_c = \diag (K_{c1}, \dots, K_{cn})$ and $T_{ci}>0, K_{ci} >0$ are the gains of the controller $i \in \mathcal{V}$.
Given $V^\ast_{di} > \sqrt{P^\ast_{li}/G^\ast_{li}}$, for all $i \in \mathcal{V}$, the equilibrium $\overline z = (\overline I_s, V^\ast_d, \overline I, \vec0, \vec0, \vec0,  u^\ast_d) \in {\mathcal{Z}}_{\mathrm{ZIP}}$ is  asymptotically stable.
\end{theorem}
\smallskip
\begin{proof}
Consider the desired closed-loop storage function
\begin{equation}
\label{eq:Sd}
S_d (\dot I_s, \dot V, \dot I, u) = S(\dot I_s, \dot V, \dot I) + \frac{1}{2} (u - u^\ast_d)^\top K_c (u - u^\ast_d),
\end{equation}
where $S$ is given by \eqref{eq:S1}. Noticing that $S_d$ is function of the entire state of the auxiliary system~\eqref{eq:plant_extended} (see Remark~\ref{rm:S}),
it is immediate to see that $S_d$ attains a minimum at the equilibrium $(\overline I_s,  V^\ast_d, \overline I, \vec0, \vec0, \vec0, u^\ast_d)$, where $\overline V = V^\ast_d$ follows from \eqref{eq:pe1} at steady state together with $\overline u = u^\ast_d$. Furthermore, $S_d$ satisfies
\begin{align}
\label{eq:dSd}
\begin{split}
\dot S_d &= - \dot V^\top \left( G^\ast_l - [V]^{-2}[P^\ast_l] \right) \dot V - \dot I^\top R \dot I \\
& \hspace{4mm}  + \upsilon_c^\top \left( K_c \left( u - u^\ast_d \right) + \dot I_s \circ V - \dot V \circ I_s \right)\\
&= - \dot V^\top \left( G^\ast_l - [V]^{-2}[P^\ast_l] \right) \dot V - \dot I^\top R \dot I - \upsilon_c^\top T_c \upsilon_c,
\end{split}
\end{align}
along the solutions to system \eqref{eq:plant_extended}. 
From the last line of~\eqref{eq:dSd} it follows that $S_d$ satisfies $\dot{S}_d \leq 0$ for all $z \in \mathcal{Z}_{\mathrm{ZIP}}$.
Then, given $\varepsilon > 0$, choose $r \in (0, \varepsilon]$ such that there exists a ball $\mathcal{B}_r(\overline z) \subset \mathcal{Z}_{\mathrm{ZIP}}$ centred in $\overline z = (\overline I_s, V^\ast_d, \overline I, \vec0, \vec0, \vec0, u^\ast_d) \in {\mathcal{Z}}_{\mathrm{ZIP}}$, i.e., 
\begin{equation*}
\mathcal{B}_r(\overline z) := \left\{ z \in {\mathcal{Z}}_{\mathrm{ZIP}}: ||z -\overline z|| \leq r \right\} \subset \mathcal{Z}_{\mathrm{ZIP}}.
\end{equation*}
 Moreover, let $\alpha$ denote the minimum value of $S_d$ on the boundary of $\mathcal{B}_r(\overline z)$, i.e., $\alpha = {\min}_{||z - \overline z|| = r} S_d(z)$. Since $S_d$ is positive definite, then $\alpha > 0$. Take $\beta \in (0,\alpha)$, then the set
 \begin{equation*}
\Omega_\beta := \left\{ z \in \mathcal{B}_r(\overline z): S_d(z) \leq \beta \right\}
\end{equation*}
is compact, positively invariant and in the interior of $\mathcal{B}_r(\overline z)$ (see \cite[Theorem 4.1]{khalil2002nonlinear}).
Let now $E$ denote the set of all points in $\Omega_\beta$ where $\dot{S}_d = 0$, i.e.,
\begin{equation*}
E := \left\{ z \in \Omega_\beta: \dot V = \vec0, \dot I = \vec0, \upsilon_c = \vec0 \right\}.
\end{equation*}
Moreover, let ${M}$ be the largest invariant set in $E$. Then, by LaSalle's invariance principle \cite[Theorem 4.4]{khalil2002nonlinear}, every solution starting in $\Omega_\beta$ approaches $M$ as $t$ approaches infinity. As a consequence, in the largest invariant set $M$, from \eqref{eq:pe4} and \eqref{eq:pe6} we obtain $\ddot I_s = \vec0$ and $\ddot I = \vec0$, respectively. Moreover, a straightforward  computation shows that in $M$ the third-time derivative of the voltage $V$ satisfies
\begin{equation*}
CV^{(3)} = -\left( T_c^{-1}[I_s]^2 + \left( G^\ast_l - [V]^{-2}[P^\ast_l] \right) \right)\ddot V,
\end{equation*}
along the solutions to system \eqref{eq:plant_extended}, implying that also $\ddot V$ is equal to zero in $M$. Consequently, from \eqref{eq:pe5} we obtain $\dot I_s=\vec0$ in $M$, and from \eqref{eq:controller} we can conclude that, in the largest invariant set $M$, $\overline u = u^\ast_d$, implying from \eqref{eq:pe1} that $V$ asymptotically converges to $V^\ast_d$. We finally conclude the proof observing from \eqref{eq:pe2} and \eqref{eq:pe3} that also $I_s$ and $I$ converge to a constant value satisfying \eqref{eq:plant_extended_ss}.
\end{proof}

\begin{figure}
\centering
\includegraphics[width=0.9\columnwidth]{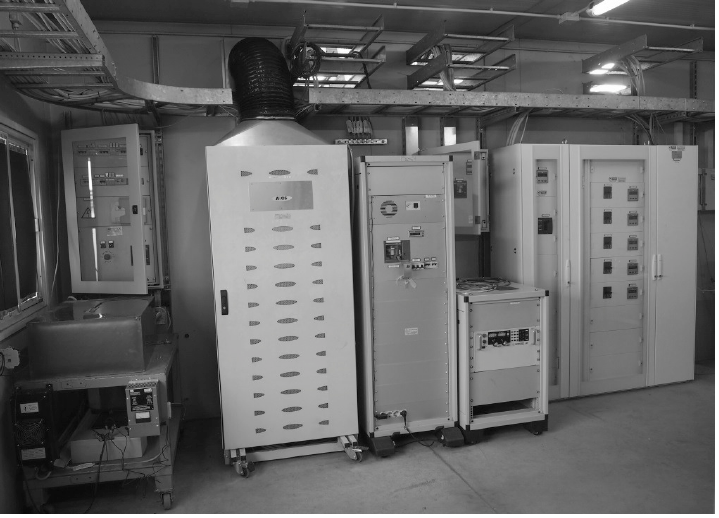}
\caption{Photo of the RSE's DC microgrid.}
\label{fig:RSE1}
\end{figure}
\begin{figure}
\centering
\includegraphics[width=0.71\columnwidth]{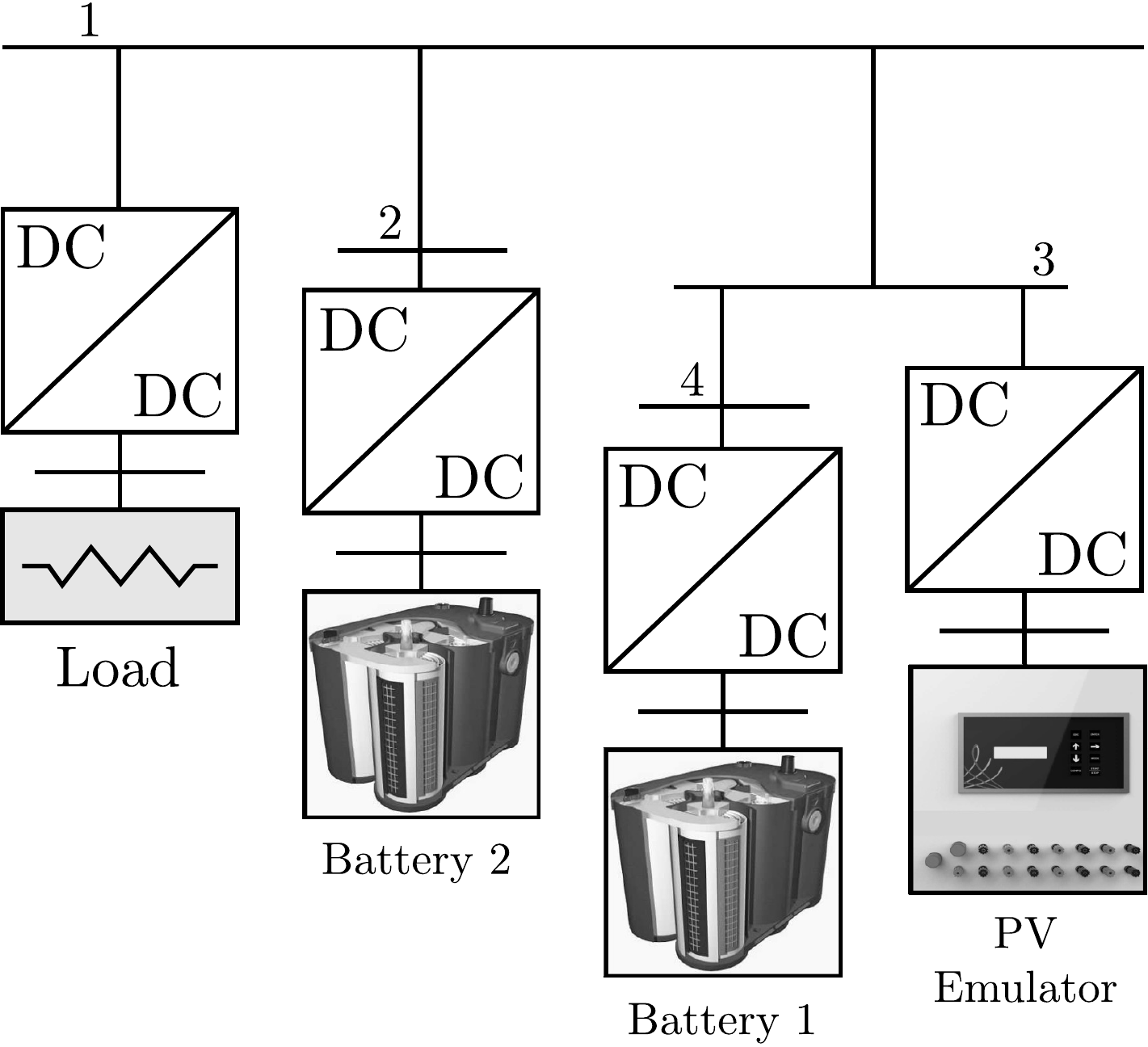}
\caption{Layout of the RSE's DC microgrid.}
\label{fig:RSE2}
\end{figure}

\begin{remark}{\bf(Robustness property)}
\label{rm:robustness}
Note that controller~\eqref{eq:controller} requires the first-time derivative of the current $I_s$ and voltage $V$, which can be estimated in finite time by implementing for instance the well known Levant's differentiator~\cite{Levant03}. Moreover, we observe that the use of $\dot I_s, \dot V$ gives robustness properties to the proposed  controller with respect to loads, lines and boost parameter uncertainty. Controller~\eqref{eq:controller} requires indeed only the knowledge of $u^\ast_d$, which depends on $V^\ast_s$.
\end{remark}

\begin{remark}{\bf(\lq ZI' loads)}
We observe that in case of only \lq ZI' loads, i.e., $P^\ast_l=\vec0$, the results developed in this section can be strengthened. The absence of constant power loads implies indeed that the passivity property of system~\eqref{eq:plant_extended} and the result of Theorem~\ref{th:1} hold in the whole set $\mathcal{Z}$. This immediately follows by noticing that $P^\ast_l=\vec0$ implies $\mathcal{Z}_{ZIP} \equiv \mathcal{Z}$.
\end{remark}
\section{Experimental Results}
\label{sec:experiments}

\begin{table}
\centering
\caption{RSE DC Microgrid parameters}
{\begin{tabular}{cccl}
\toprule			
Symbol				&Value &Unit			&Description\\
\midrule
{${V_{s2}},{V_{s4}}$}		&{278} &{\SI{}{\volt}}		&{Batteries voltage source}\\
${V^{\ast}_d}$		&380 &\SI{}{\volt}		&Desired voltage\\
${R_{12}}$				&250 &\SI{}{\milli\ohm}	& Line resistance 1-2\\
${R_{13}}$				&39 &\SI{}{\milli\ohm}	& Line resistance 1-3\\
${R_{34}}$				&250 &\SI{}{\milli\ohm}	& Line resistance 3-4\\
${L_{12}}$				&140 &\SI{}{\micro\henry}	& Line inductance 1-2\\
${L_{13}}$				&86 &\SI{}{\micro\henry}	& Line inductance 1-3\\
${L_{34}}$				&140 &\SI{}{\micro\henry}	& Line inductance 3-4\\
${C_{2}},{C_{4}}$				&6.8 &\SI{}{\milli\farad}	& Output capacitances\\
${L_{s2}},{L_{s4}}$				&1.12 &\SI{}{\milli\henry}	& Input inductances\\
$f_{\mathrm sw}$		&4 &\SI{}{\kilo\hertz}		& Switching frequency\\
\bottomrule
\end{tabular}}
\label{tab:parameters1}
\end{table}

In order to validate the proposed control scheme, experimental tests have been performed on the DC microgrid test facility at RSE. The electrical scheme, a photo and the layout of the setup are shown in Figs.~\ref{fig:microgrid2}--\ref{fig:RSE2}.
The RSE's DC microgrid is unipolar with a nominal voltage of 380 V and includes a ZIP load, with a maximum power of 30 kW at 400 V, a DC generator with a maximum power of 30 kW, which emulates a PV plant, and two storage devices, based on high temperature NaNiCl batteries, each of them with an energy of 18 kWh and a maximum power of 30 kW for 10~s. The batteries are connected to the DC network through 35~kW bidirectional boost converters. All the other parameters of the RSE's DC microgrid are reported in Table~\ref{tab:parameters1}.

In order to regulate the voltages $V_2$ and $V_4$ at the nodes 2 and 4 towards the corresponding desired value ${V^\ast_d}=$ \SI{380}{V}, the control strategy proposed in Section \ref{sec:control} (with $T_c =$ \num{1e7} and $K_c =$ \num{1e9}) is implemented through dSpace controllers.
The currents $I_{l1}(V_1)$ and $I_{l3}(V_3)$ demanded by the load and generated by the PV emulator are treated as disturbances.
In the following, we arbitrarily assume the passive sign convention\footnote{$I_{l1}(V_1) \geq 0$, $I_{l3}(V_3)\leq 0$, $I_{s1}, I_{s2} > 0$ ($I_{s1}, I_{s2} < 0$) if the batteries charge (discharge).}.

In the first scenario the system is in a steady state condition with zero power absorbed by the load or provided by the generator. Each battery converter regulates its output voltage at the desired value ${V^\ast_d}=$ \SI{380}{V}. At the time instant $t =$ \SI{5}{\second} the load (see Fig. \ref{fig:results1}) or the PV emulator (see Fig.~\ref{fig:results2}) absorbs/generates \SI{20}{\kilo\watt} until the time instant $t =$ \SI{45}{\second}. 
 From Fig. \ref{fig:results1} and in Fig. \ref{fig:results2}, one can observe that, after a transient due to the load/generator variations, the system exhibits a stable performance. This clearly shows the robustness of the proposed controller with respect to the  system uncertainty.

In the second scenario the system is in a steady state condition with a constant power equal to \SI{20}{\kilo\watt} provided by the generator. Each battery converter regulates its output voltage at the desired value ${V^\ast_d}=$ \SI{380}{V}. At the time instant $t =$ \SI{5}{\second} the desired value ${V^\ast_{d2}}$ is changed to \SI{375}{V} and at the time instant $t =$ \SI{45}{\second} also the desired value ${V^\ast_{d4}}$ is changed to \SI{375}{V} (dashed line). From Fig. \ref{fig:results4}, one can observe that the system exhibits a stable performance tracking the new desired voltage value. A similar scenario is illustrated in Fig. \ref{fig:results3}. Tracking capabilities are generally essential to couple voltage controllers with higher-level control schemes that modifies the voltage reference of each node, in order to achieve power sharing among the nodes of the microgrid. 

Finally, we note that in the discussed scenarios, only the voltages $V_2$ and $V_4$ are controlled and the deviations from the desired value during the load and generator variations are less than 4{\%}. In the uncontrolled nodes, the deviations of the voltages $V_1$ and $V_3$ from the desired value are less than 7{\%}. These deviations are due to the line impedances between the controlled and uncontrolled nodes.

\begin{figure}
\vspace{0.3cm}
\centering
\includegraphics[width=\columnwidth]{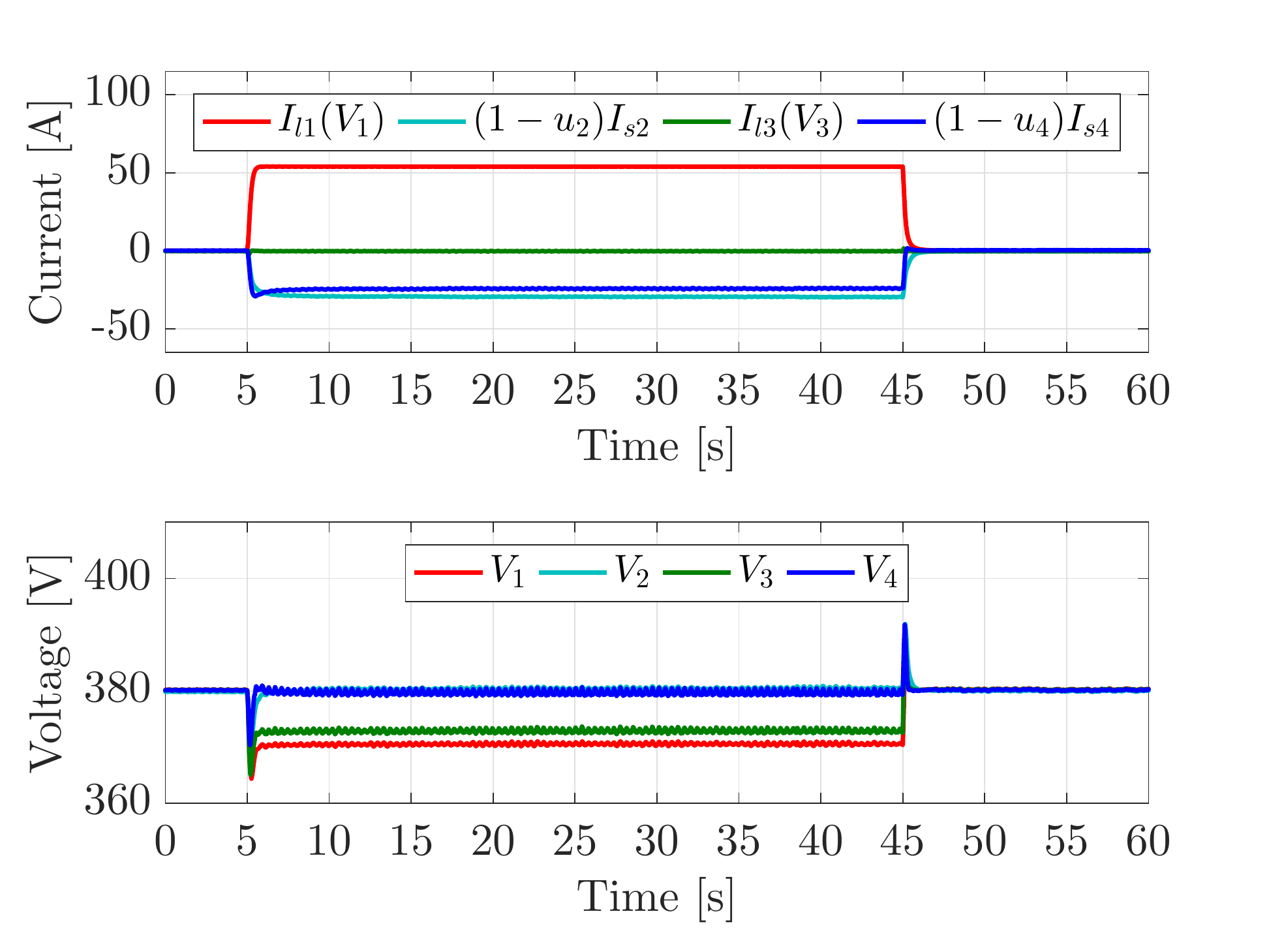}
\caption{Scenario 1: closed-loop system performance with a step load variation of  \SI{20}{\kilo\watt}.}
\label{fig:results1}
\end{figure}
\begin{figure}
\vspace{0.5cm}
\centering
\includegraphics[width=\columnwidth]{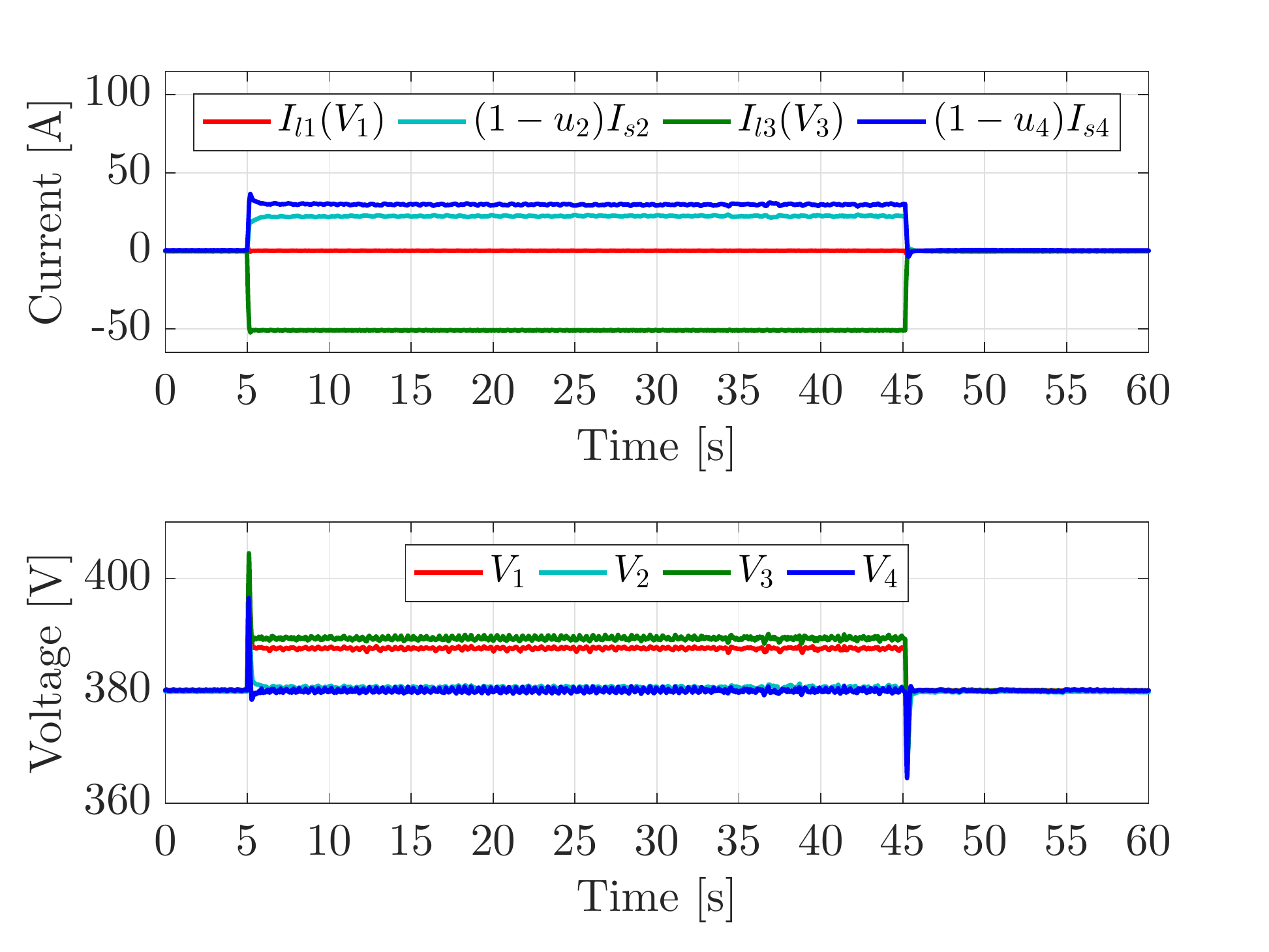}
\caption{Scenario 1: closed-loop system performance with a step generator variation of \SI{20}{\kilo\watt}.}
\label{fig:results2}
\end{figure}
\begin{figure}
\vspace{0.5cm}
\centering
\includegraphics[width=\columnwidth]{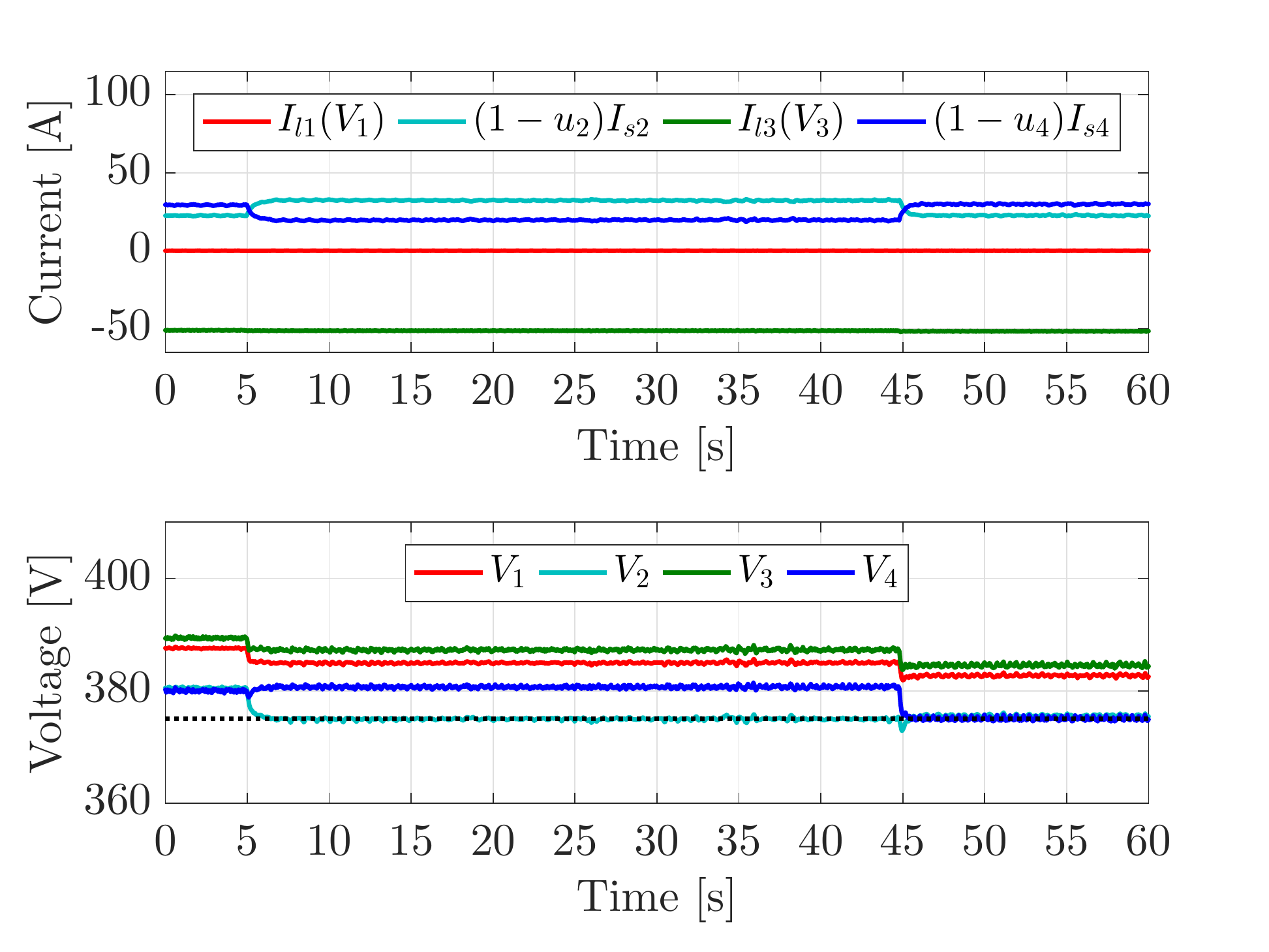}
\caption{Scenario 2: closed-loop system performance with a step reference variation of -\SI{5}{\volt}.}
\label{fig:results4}
\end{figure}
\begin{figure}
\vspace{0.5cm}
\centering
\includegraphics[width=\columnwidth]{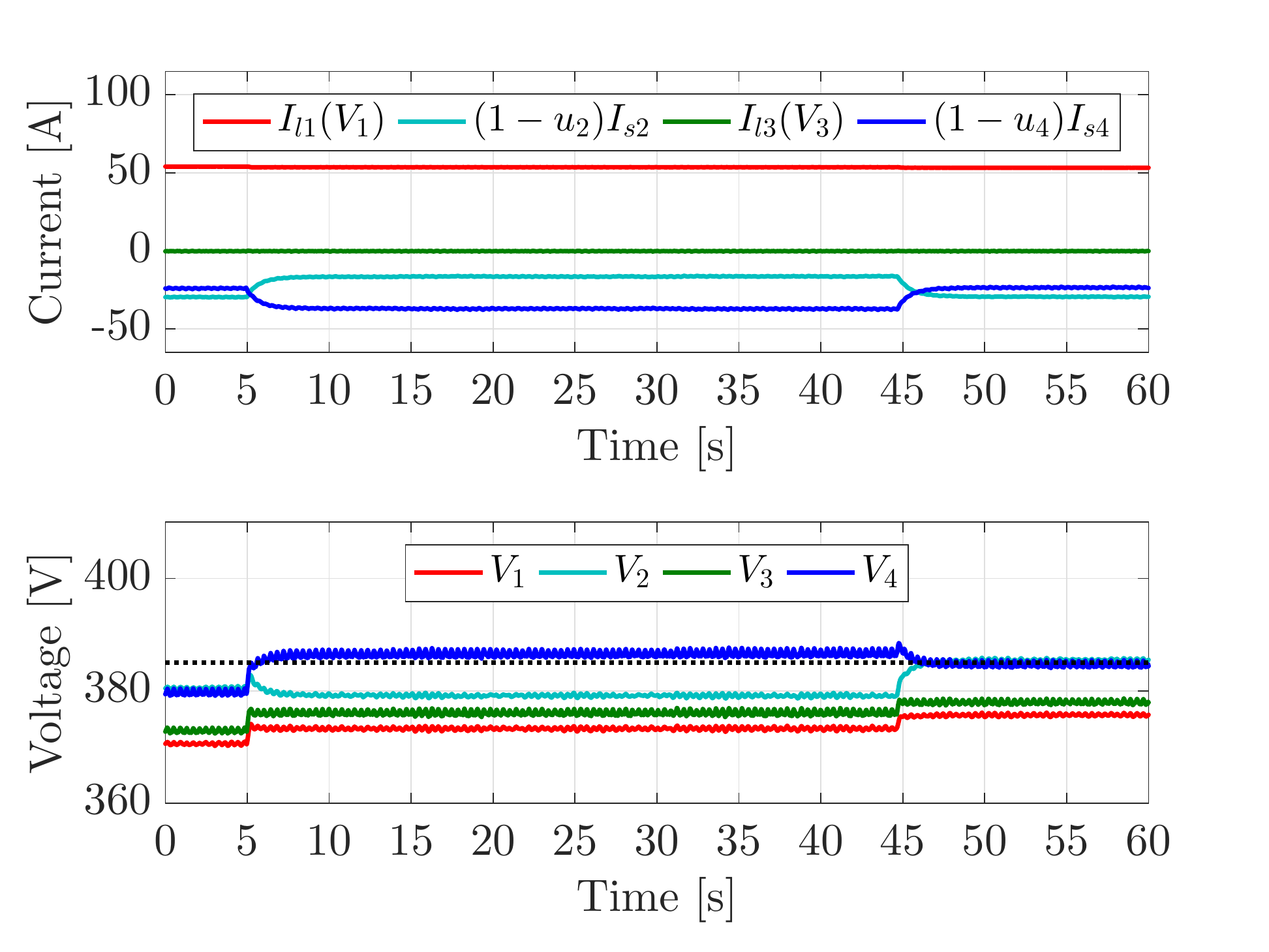}
\caption{Scenario 2: closed-loop system performance with a step reference variation of +\SI{5}{\volt}.}
\label{fig:results3}
\end{figure}
%
%
%
%
\section{CONCLUSIONS}
\label{sec:conclusions}
In this paper a decentralized passivity-based control scheme is designed to regulate the voltage of a DC microgrid through \emph{boost} power converters. 
Using a Krasovskii-type storage function, a (local) passivity property for the considered DC microgrid is established.
More precisely, the integrated input port-variable is used to shape the closed loop storage function.
Convergence to the desired equilibrium is proven in presence of the so-called \lq ZIP' (constant impedance \lq Z', constant current \lq I' and constant power \lq P') loads, showing robustness with respect to system parameter uncertainties.
The proposed control scheme is validated through experimental tests on a real DC microgrid, showing excellent closed-loop performances.
\bibliographystyle{IEEEtran}

\bibliography{ECC19ref}

\begin{thebibliography}{10}
\providecommand{\url}[1]{#1}
\csname url@samestyle\endcsname
\providecommand{\newblock}{\relax}
\providecommand{\bibinfo}[2]{#2}
\providecommand{\BIBentrySTDinterwordspacing}{\spaceskip=0pt\relax}
\providecommand{\BIBentryALTinterwordstretchfactor}{4}
\providecommand{\BIBentryALTinterwordspacing}{\spaceskip=\fontdimen2\font plus
\BIBentryALTinterwordstretchfactor\fontdimen3\font minus
  \fontdimen4\font\relax}
\providecommand{\BIBforeignlanguage}[2]{{%
\expandafter\ifx\csname l@#1\endcsname\relax
\typeout{** WARNING: IEEEtran.bst: No hyphenation pattern has been}%
\typeout{** loaded for the language `#1'. Using the pattern for}%
\typeout{** the default language instead.}%
\else
\language=\csname l@#1\endcsname
\fi
#2}}
\providecommand{\BIBdecl}{\relax}
\BIBdecl

\bibitem{ACKERMANN2001195}
T.~Ackermann, G.~Andersson, and L.~S{\"o}der, ``Distributed generation: a
  definition,'' \emph{Electric Power Systems Research}, vol.~57, no.~3, pp.
  195--204, Apr. 2001.

\bibitem{hatziargyriou2014microgrids}
N.~Hatziargyriou, \emph{Microgrids: architectures and control}.\hskip 1em plus
  0.5em minus 0.4em\relax John Wiley \& Sons, 2014.

\bibitem{doi:10.1080/00207179.2015.1104555}
D.~Efimov, J.~Schiffer, and R.~Ortega, ``Robustness of delayed multistable
  systems with application to droop-controlled inverter-based microgrids,''
  \emph{International Journal of Control}, vol.~89, no.~5, pp. 909--918, Jan.
  2016.

\bibitem{7500071}
J.~W. {Simpson-Porco}, F.~{D{\"o}rfler}, and F.~{Bullo}, ``Voltage
  stabilization in microgrids via quadratic droop control,'' \emph{IEEE
  Transactions on Automatic Control}, vol.~62, no.~3, pp. 1239--1253, Mar.
  2017.

\bibitem{GUI2018551}
Y.~Gui, B.~Wei, M.~Li, J.~M. Guerrero, and J.~C. Vasquez, ``Passivity-based
  coordinated control for islanded {AC} microgrid,'' \emph{Applied Energy},
  vol. 229, pp. 551--561, Nov. 2018.

\bibitem{C_Cucuzzella_18_3}
M.~Cucuzzella, S.~Trip, A.~Ferrara, and J.~M.~A. Scherpen, ``{Cooperative
  Sliding Mode Voltage Control in AC Microgrids},'' in \emph{Proc. IEEE 57th
  Conf. Decision Control}, Miami Beach, FL, USA, Dec. 2018.

\bibitem{7268934}
T.~Dragi{\v c}evi{\'c}, X.~Lu, J.~C. Vasquez, and J.~M. Guerrero, ``{DC
  Microgrids--Part I: A Review of Control Strategies and Stabilization
  Techniques},'' \emph{IEEE Transactions on Power Electronics}, vol.~31, no.~7,
  pp. 4876--4891, July 2016.

\bibitem{PLANAS2015726}
E.~Planas, J.~Andreu, J.~I. G{\'a}rate, I.~M. de~Alegr{\'\i}a, and E.~Ibarra,
  ``{AC and DC technology in microgrids: A review},'' \emph{Renewable and
  Sustainable Energy Reviews}, vol.~43, pp. 726--749, Mar. 2015.

\bibitem{DePersis2016}
C.~{De Persis}, E.~R. Weitenberg, and F.~D{\"o}rfler, ``A power consensus
  algorithm for {DC} microgrids,'' \emph{Automatica}, vol.~89, pp. 364--375,
  Mar. 2018.

\bibitem{han8026170}
R.~Han, L.~Meng, J.~M. Guerrero, and J.~C. Vasquez, ``Distributed nonlinear
  control with event-triggered communication to achieve current-sharing and
  voltage regulation in {DC} microgrids,'' \emph{IEEE Transactions on Power
  Electronics}, vol.~33, no.~7, pp. 6416--6433, July 2018.

\bibitem{cucuzzella2017robust}
M.~Cucuzzella, S.~Trip, C.~{De Persis}, X.~Cheng, A.~Ferrara, and A.~van~der
  Schaft, ``{A Robust Consensus Algorithm for Current Sharing and Voltage
  Regulation in {DC} Microgrids},'' \emph{IEEE Transactions on Control Systems
  Technology}, 2018.

\bibitem{J_Trip2018}
S.~Trip, M.~Cucuzzella, X.~Cheng, and J.~Scherpen, ``{Distributed Averaging
  Control for Voltage Regulation and Current Sharing in {DC} Microgrids},''
  \emph{IEEE Control Systems Letters}, vol.~3, no.~1, pp. 174--179, Jan. 2019.

\bibitem{StrehlePfeiferKrebs2019_1000090285}
F.~Strehle, M.~Pfeifer, S.~Krebs, and S.~Hohmann,
  ``\BIBforeignlanguage{english}{A scalable port-hamiltonian approach to
  plug-and-play voltage stabilization in dc microgrids},'' in
  \emph{\BIBforeignlanguage{english}{IFAC-PapersOnLine (submitted)}}, 2019.

\bibitem{1323174}
D.~Jeltsema and J.~M.~A. Scherpen, ``Tuning of passivity-preserving controllers
  for switched-mode power converters,'' \emph{IEEE Transactions on Automatic
  Control}, vol.~49, no.~8, pp. 1333--1344, Aug. 2004.

\bibitem{Moreno2017}
J.~{Moreno-Valenzuela} and O.~{Garcia-Alarcon}, ``On control of a boost dc-dc
  power converter under constrained input,'' \emph{Complexity}, Jan. 2017.

\bibitem{7983406}
M.~S. Sadabadi, Q.~Shafiee, and A.~Karimi, ``Plug-and-play robust voltage
  control of {DC} microgrids,'' \emph{IEEE Transactions on Smart Grid}, vol.~9,
  no.~6, pp. 6886--6896, Nov. 2018.

\bibitem{8618882}
A.~{Martinelli}, P.~{Nahata}, and G.~{Ferrari-Trecate}, ``{Voltage
  Stabilization in MVDC Microgrids Using Passivity-Based Nonlinear Control},''
  in \emph{Proc. IEEE 57th Conf. Decision Control}, Miami Beach, FL, USA, Dec.
  2018, pp. 7022--7027.

\bibitem{Iovine}
A.~{Iovine} and F.~{Mazenc}, ``Bounded control for dc/dc converters:
  Application to renewable sources,'' in \emph{2018 IEEE Conference on Decision
  and Control (CDC)}, Dec. 2018, pp. 3415--3420.

\bibitem{krasovskiicertain}
N.~Krasovskii, \emph{Certain Problems of the Theory of Stability of Motion [in
  Russian], Fizmatgiz, Moscow, 1959}.\hskip 1em plus 0.5em minus 0.4em\relax
  English translation by Stanford University Press, 1963.

\bibitem{7846443}
K.~C. Kosaraju, R.~Pasumarthy, N.~M. Singh, and A.~L. Fradkov, ``Control using
  new passivity property with differentiation at both ports,'' in \emph{2017
  Indian Control Conference (ICC)}, Guwahati, India, Jan. 2017, pp. 7--11.

\bibitem{2018arXiv181102838C}
K.~C. {Kosaraju}, M.~{Cucuzzella}, R.~{Pasumarthy}, and J.~M.~A. {Scherpen},
  ``{Differentiation and Passivity for Control of Brayton-Moser Systems},''
  \emph{arXiv preprint: 1811.02838}, Nov. 2018.

\bibitem{7892743}
D.~Ronchegalli and R.~Lazzari, ``{Development of the control strategy for a
  direct current microgrid: A case study},'' in \emph{2016 AEIT International
  Annual Conference (AEIT)}, Capri, Italy, Oct. 2016.

\bibitem{J_Cucuzzella_18}
M.~Cucuzzella, R.~Lazzari, S.~Trip, S.~Rosti, C.~Sandroni, and A.~Ferrara,
  ``Sliding mode voltage control of boost converters in $\mathrm{DC}$
  microgrids,'' \emph{Control Engineering Practice}, vol.~73, pp. 161--170,
  Apr. 2018.

\bibitem{khalil2002nonlinear}
H.~Khalil, \emph{{Nonlinear Systems}}, 3rd~ed.\hskip 1em plus 0.5em minus
  0.4em\relax Prentice Hall, 2002.

\bibitem{Levant03}
A.~Levant, ``Higher-order sliding modes, differentiation and output-feedback
  control,'' \emph{Int. J. Control}, vol.~76, no. 9-10, pp. 924--941, Jan.
  2003.

\end{thebibliography}

\end{document}